\documentclass[a4paper,twoside,reqno]{amsart}

\usepackage[pagewise]{lineno}
\usepackage[utf8]{inputenc}

\usepackage{authblk}
\usepackage{hyperref}
\hypersetup{urlcolor=blue, citecolor=red}
\usepackage{amsmath,amsthm,amssymb,xcolor,bbm,mathrsfs}
\usepackage[english]{babel}
\usepackage{fancyhdr}
\makeatletter
\@namedef{subjclassname@1991}{2020 Mathematics Subject Classification}
\makeatother

\title{Application of the Deffuant model\\ in money exchange}
\author{Hsin-Lun Li}

\date{}
\email{hsinlunl@asu.edu}

\theoremstyle{definition}
\newtheorem{theorem}{Theorem}
\newtheorem{definition}[theorem]{Definition}
\newtheorem{lemma}[theorem]{Lemma}

\begin{document}

\allowdisplaybreaks

\thispagestyle{firstpage}
\maketitle
\begin{center}
    Hsin-Lun Li
    \centerline{$^1$National Sun Yat-sen University, Kaohsiung 804, Taiwan}
\end{center}
\medskip
\begin{abstract}
    A money transfer involves a buyer and a seller. A buyer buys goods or services from a seller. The money the buyer decreases is the same as that the seller increases. At each time step, a pair of socially connected agents are selected and transact in agreed money. We evolve the Deffuant model to a money exchange system and study circumstances under which asymptotic stability holds, or equal wealth can be achieved. 
\end{abstract}
\section{Introduction}
The Deffuant model is one of the popular models in opinion dynamics. The original Deffuant model consists of a finite number of agents whose opinion is a number in $[0,1].$ Based on a fixed social relationship among all agents, two socially connected agents are selected at each time step and interact if and only if their opinion distance does not exceed some confidence threshold $\epsilon$. The interacting mechanism goes as follows:
\begin{align}\begin{aligned}\label{Deffuant model}
    &x_i(t+1)=x_i(t)+\mu(x_j(t)-x_i(t))\mathbbm{1}\{|x_i(t)-x_j(t)|\leq \epsilon\}\\
    &x_j(t+1)=x_j(t)+\mu(x_i(t)-x_j(t))\mathbbm{1}\{|x_i(t)-x_j(t)|\leq \epsilon\}
    \end{aligned}
\end{align}
for $x_i(t)\in [0,1]$ the opinion of agent $i$ at time $t$ and $\mu\in (0,1/2]$ the convergence parameter. Namely, agents $i$ and $j$ move equally toward each other. The Hegselmann-Krause model is another popular model in opinion dynamics. Without a social relationship, an agent updates its opinion by taking the average opinion of its opinion neighbors, those whose opinion differs by at most some confidence threshold. The authors in \cite{lanchier2020probability} derive a nontrivial lower bound for the probability of consensus under the Deffuant model. The author in \cite{mHK2}, the sequel to \cite{mHK}, introduces a variant of the Hegselmann-Krause model and argues that it covers not only the Hegselmann-Krause model but also the Deffuant model. With similar methods in \cite{lanchier2020probability}, the authors in \cite{lanchier2022consensus} get a nontrivial lower bound for the probability of consensus under the Hegselmann-Krause model.

A money transfer involves a buyer and a seller. The buyer pays for goods or services offered by the seller. Therefore, the buyer decreases the same money as the seller increases. Considering a finite set of agents, say two agents are \emph{socially connected} if there is a social relationship, such as trust, between them. Say $m_i$ is the money of agent $i$. $m_i<0$ if agent $i$ is in debt. It is realistic to set a lower bound for $m_i$ since no one is allowed to borrow money endlessly. Also, we assume the money transfer system is conservative. Namely, influx equals efflux, therefore the sum of all agents' money constant over time. Before depicting the money transfer system evolved from the Deffuant model, we introduce the following terms. Let $[n]=\{1,2,\ldots,n\}$ denote the collection of all agents. 

\begin{definition}{\rm
A \emph{social graph} at time $t$, $G(t)=([n], E(t))$, is an undirected graph with vertex set and edge set,
$$[n]\ \hbox{and}\ E(t)=\{(i,j)\in [n]^2: i\neq j\ \hbox{and}\  \hbox{vertices $i$ and $j$ are socially connected}\}.$$
A \emph{social graph for update} at time $t$, $\Tilde{G}(t)=([n],\Tilde{E}(t))$, is a subgraph of the social graph at time $t$ for money update.}
\end{definition}

\begin{definition}{\rm
A graph is \emph{$\delta$-trivial} if any two vertices in the graph are at a distance of at most $\delta$ apart. }
\end{definition}

The Deffuant model evolves into a money transfer system with the following setting: Let $m_i(0),\ i\in[n]$ and $\mu(t)$, $t\geq 0$ be independent continuous real-valued random variables. Agents $i$ and $j$ are selected at time $t$ and interact if and only if they are socially connected. The interaction mechanism is as follows.
\begin{align}\begin{aligned}\label{money transfer system}
   & m_i(t+1)=m_i(t)+\mu(t)\big(m_j(t)-m_i(t)\big)\mathbbm{1}\{(i,j)\in E(t)\}  \\
   &  m_j(t+1)=m_j(t)+\mu(t)\big(m_i(t)-m_j(t)\big)\mathbbm{1}\{(i,j)\in E(t)\}
   \end{aligned}
\end{align}
where
$$\begin{array}{l}
   \displaystyle m_i(t)\geq -d_i\ \hbox{for}\ d_i>0\ \hbox{constant and for all}\ t\geq 0,\vspace{2pt}\\
   \displaystyle \sum_{i\in [n]}m_i=C\ \hbox{for}\ C\ \hbox{constant}.
\end{array}$$
 Observe that $d_i$ symbolizes the credibility of agent $i$ and $\mu(t)$ corresponds to the behavior that the money two agents agree to transact given that no one is out of credibility. Say agent $i$ is richer than agent $j$ if $m_i\geq m_j$, and poorer than agent $j$ if $m_i\leq m_j$. Given $m_i(t)\leq m_j(t)$, $m_i(t+1)\leq m_j(t+1)$ if $\mu(t)\leq 1/2$, and $m_i(t+1)\geq m_j(t+1)$ if $\mu(t)>1/2$. The former indicates that the richer agent remains in its richer status, whereas the latter indicates the richer agent becomes poorer, after the transaction, whether it is a buyer or a seller at time $t$. The money transfer system evolved from the Deffuant model differs from the money transfer system in \cite{dragulescu2000statistical, kiyotaki1993search, molico2006distribution} regarding the conservation of money. 

 Let $U_t=\{(i,j):\hbox{agent $i$ and $j$ are selected at time $t$}\ \hbox{and}\ \mu(t)\neq 0\}$,\ $t\geq0$ be independent and identically distributed random variables with a support $$S\subset\big\{\{(i,j)\}:i,j\in[n]\ \hbox{and}\ i\neq j\big\}\ \hbox{and}\ \Tilde{E}(t)=U_t\cap E(t).$$ In other words, $U_t$ indicates a possible pair of candidates for money transfer. Assume that $(\Omega,\mathscr{F},P)$ is a probability space for $\mathscr{F}\subset\mathscr{P}(\Omega)$ a $\sigma$-algebra and $P$ a probability measure. Denote $E(G)$ as the edge set of graph $G$. Given a pair of agents in transaction, we consider the following conditions in \eqref{money transfer system}:
 \begin{itemize}
     \item $\mu\in(0,1)$,
     \item $\mu>1$ and 
     \item $\mu<0$ 
 \end{itemize} 
corresponding to the richer agent whose money is
\begin{itemize}
    \item at least the same as the current money of the poorer agent after the transaction,
    \item at most the same as the current money of the poorer agent after the transaction, and
    \item at least the same as its current money after the transaction.
\end{itemize}

\section{Main results}

It turns out that equal wealth can be achieved if all agents are willing and possible to transact and the social graph is connected infinitely many times. Namely, all agents eventually have the average money of the total money in the system.  
\begin{theorem}\label{thm: equal wealth}
Assume that $\sup_{t\geq 0}|\mu(t)-1/2|<1/2$, the social graph connected infinitely many times and $\bigcup_{a\in S}a\supset \binom{[n]}{2}$. Then, equal wealth can be achieved eventually.
\end{theorem}
It follows that equal wealth can be achieved for a constant connected social graph if all socially connected agents are willing and possible to transact. Let $m_{(i)}$ be the $i$th smallest among $m_1,m_2,\ldots,m_n$. Theorem~\ref{thm: order} shows conditions under which $m_{(i)}$ is asymptotically stable. 

\begin{theorem}\label{thm: order}
 Assume that $\inf_{t\geq 0}|\mu(t)-1/2|\geq 1/2$, $\bigcup_{a\in S}a\supset \binom{[n]}{2}$ and the social graph complete infinitely many times. Then, $m_{(i)}$ is asymptotically stable for all $i\in [n]$ almost surely.
\end{theorem}
$\inf_{t\geq 0}|\mu(t)-1/2|\geq 1/2$ indicates that the richer agent can be either much richer or poorer than the poor agent after a transaction.

\section{The model}
The crucial part of proving Theorem~\ref{thm: equal wealth} is to find a nonnegative nonincreasing function and construct an inequality involving the current money and the updated money. Then, we are able to show circumstances under which asymptotic stability holds in the money transfer system.
\begin{lemma}\label{key}
    Let $Z(t)=\sum_{i,j\in [n]}\big(m_i(t)-m_j(t)\big)^2$. Then, $Z(t)$ is monotone with respect to $t$. In particular,
    $$Z(t)-Z(t+1)= 2n(\frac{1}{\mu(t)}-1)\mathbbm{1}\{\mu(t)\neq 0\}\sum_{i\in [n]}\big(m_i(t)-m_i(t+1)\big)^2.$$
\end{lemma}

\begin{proof}
 Let $\Tilde{E}(t)=\{(p,q)\}$, $\mu=\mu(t)$, $m_i=m_i(t)$ and $m_i^\star=m_i(t+1)$ for all $i\in [n].$ Clearly, $Z(t)-Z(t+1)=0$ for $\mu=0$. For $\mu\neq 0$, we have
 \begin{align*}
     &(m_p-m_q)^2-(m_p^\star-m_q^\star)^2=[1-(1-2\mu)^2](m_p-m_q)^2\\
     &\hspace{2cm}=(2-2\mu)2\mu/\mu^2(m_p-m_p^\star)^2=4\frac{(1-\mu)}{\mu}(m_p-m_p^\star)^2, \\ 
     &(m_p-m_i)^2-(m_p^\star-m_i)^2=(m_p-m_p^\star)^2+2(m_p-m_p^\star)(m_p^\star-m_i),\\
     &(m_q-m_i)^2-(m_q^\star-m_i)^2=(m_q-m_q^\star)^2+2(m_q-m_q^\star)(m_q^\star-m_i).
 \end{align*}
 Since $m_p-m_p^\star=-(m_q-m_q^\star),$
 \begin{align*}
     &(m_p-m_i)^2-(m_p^\star-m_i)^2+(m_q-m_i)^2-(m_q^\star-m_i)^2\\
     &\hspace{1cm}=2(m_p-m_p^\star)^2+2(m_p-m_p^\star)(m_p^\star-m_q^\star)\\
     &\hspace{1cm}=[2+2(1-2\mu)/\mu](m_p-m_p^\star)^2=2(\frac{1}{\mu}-1)(m_p-m_p^\star)^2.
 \end{align*}
In particular,
\begin{align*}
    &Z(t)-Z(t+1)=\sum_{i,j\in [n]}[(m_i-m_j)^2-(m_i^\star-m_j^\star)^2]\\
    &\hspace{0.5cm}=2\bigg\{(m_p-m_q)^2-(m_p^\star-m_q^\star)^2\\
    &\hspace{1.3cm}+\sum_{i\in [n]-\{p,q\}}\big[(m_p-m_i)^2-(m_p^\star-m_i)^2+(m_q-m_i)^2-(m_q-m_i)^2\big]\bigg\}\\
    &\hspace{0.5cm}=2\big[4(\frac{1}{\mu}-1)+2(n-2)(\frac{1}{\mu}-1)\big](m_p-m_p^\star)^2=4n(\frac{1}{\mu}-1)(m_p-m_p^\star)^2.
\end{align*}
\end{proof}

\begin{lemma}\label{lemma:asymptotic stability}
Assume that $\sup_{t\geq 0}|\mu(t)-1/2|<1/2$. Then, all components of $\Tilde{G}$ are $\delta$-trivial after some finite time for all $\delta>0$ almost surely.   
\end{lemma}

\begin{proof}
    Assume by contradiction that this is not the case. Then, there is $(t_k)_{k\geq 0}$ increasing with $\Tilde{G}(t_k)$ $\delta$-nontrivial for all $k\geq 0$. $R=\sup_{t\geq 0}|\mu(t)-1/2|$ implies $1/2-R \leq \mu(t) \leq 1/2+R$, therefore $1/\mu(t)-1\geq \frac{1/2-R}{1/2+R}$ for all $t\geq 0$. It follows from Lemma~\ref{key} that $$Z(0)\geq Z(0)-Z(s)=\sum_{t=0}^{s-1}[Z(t)-Z(t+1)]\ \hbox{for all}\ s\geq 1.$$
    Letting $s\to\infty,$
    $$\infty>Z(0)\geq\sum_{t\geq 0}[Z(t)-Z(t+1)]>\sum_{k\geq 0}4n\frac{(1/2-R)^3}{1/2+R}\delta^2=\infty,\ \hbox{a contradiction.}$$
\end{proof}

\begin{proof}[\bf Proof of Theorem~\ref{thm: equal wealth}]
    Due to finiteness of the social graph, it connected infinitely many times implies there is a graph $H$ connected infinitely many times, saying $(t_k)_{k\geq 0}$ the timepoints. For $(i,j)\in E(H)$, let $B_t$ be the event that $(i,j)\in U_t$. Because of $\bigcup_{a\in S}a\supset\binom{[n]}{2}$ and support $S$ finite, via second Borel–Cantelli lemma, $$\sum_{k\geq 0}P(B_{t_k})\geq\sum_{k\geq 0}\min_{a\in S}P(U_{t_k}=a)\geq\sum_{k\geq 0}\min_{a\in S}P(U_0=a)=\infty$$ implies $B_{s_\ell},\ \ell\geq 0$ hold for $(s_{\ell})_{\ell\geq 0}\subset (t_k)_{k\geq 0}$, therefore $(i,j)\in \Tilde{E}(s_\ell).$ Via Lemma~\ref{lemma:asymptotic stability}, vertices $i$ and $j$ have the same money eventually. Graph $H$ connected implies all vertices achieve the same money.
\end{proof}

\begin{proof}[\bf Proof of Theorem~\ref{thm: order}]
     We claim that $\lim_{t\to \infty}|\mu(t)-1/2|=1/2$. Let $A_t=\sup\big\{|m_i(t)-m_j(t)|:i,j\in [n],\ -d_i<m_i(t)\ \hbox{and}\ -d_j<m_j(t)\big\}.$ For $c=0$ or $-\infty$, $A_t=c$ implies $A_s=c$ for some $t$ and for all $s\geq t.$ It is clear that $P(A_0=c)=0$. If $A_t>0$ for all $t\geq 0,$ then $A_t$ is nondecreasing with respect to $t$. Assume that $\limsup_{t\to\infty}|\mu(t)-1/2|>1/2$. Then, there are $\delta>0$ and $(t_k)_{k\geq 0}$ increasing such that $|\mu(t_k)-1/2|\geq 1/2+\delta$ for all $k\geq 0$. Let $i_t=\arg\max_{i}\{m_i(t):i\in [n]\ \hbox{and}\ -d_i<m_i(t)\}$, $j_t=\arg\min_{i}\{m_i(t):i\in [n]\ \hbox{and}\ -d_i<m_i(t)\}$ and $E_t$ be the event that agents $i_{t_0}$ and $j_{t_0}$ transact at time $t$. Due to finite support $S$, $\bigcup_{a\in S}a\supset \binom{[n]}{2}$ and the second Borel-Cantelli lemma, we have $\sum_{k\geq 0}P(E_{t_k})=\infty$, which implies $E_t$ holds for infinitely many $t\in (t_k)_{k\geq 0}$. Say $\ell_1\in (t_k)_{k\geq 0}$ is the earliest time that $E_{\ell_1}$ holds. Then, $A_{\ell_1+1}\geq |1-2\mu(\ell_1)|A_{\ell_1} \geq (1+2\delta)A_{\ell_1}$. Repeating the above process, there is $(\ell_k)_{k\geq 1}\subset (t_k)_{k\geq 0}$ increasing such that $A_{\ell_k+1}\geq (1+2\delta)^k A_{\ell_1}\to\infty$ as $k\to\infty$, a contradiction. This completes the proof.
\end{proof}

\section*{Acknowledgment}
The author is funded by the National Science and Technology Council in Taiwan.

\end{document}